\newif\ifabstract
\abstracttrue
\newif\iffull
\ifabstract \fullfalse \else \fulltrue \fi

\iffull
\documentclass[11pt,letterpaper]{article}
\usepackage[margin=1in]{geometry}
\usepackage{amsthm}
\usepackage{amsmath}

\newtheorem{theorem}{Theorem}[section]
\newtheorem{corollary}{Corollary}[section]

\fi
\ifabstract
\documentclass{cccg24}
\usepackage{graphicx,amssymb,amsmath}

\makeatletter
\renewcommand{\paragraph}{%
  \@startsection{paragraph}{4}%
  {\z@}{2ex \@plus 1ex \@minus .2ex}{-1em}%
  {\normalfont\normalsize\bfseries}%
}
\makeatother
\fi

\usepackage{enumitem}

\usepackage{graphicx}
\usepackage{todonotes}
\usepackage{url}
\usepackage{CJKutf8} 

\usepackage[hidelinks]{hyperref}
\usepackage[capitalise]{cleveref}

\graphicspath{{./graphics/}}

\newtheorem{corollary}[theorem]{Corollary}

\def\defn#1{\textit{\textbf{\boldmath #1}}\index{#1}}

\title{Slant/Gokigen Naname is NP-complete, and Some Variations are in P}
\author{
Jayson Lynch\thanks{MIT-CSAIL, {\tt jaysonl@mit.edu}}
\and
Jack Spalding-Jamieson\thanks{David R. Cheriton School of Computer Science, University of Waterloo, {\tt jacksj@uwaterloo.ca}}
}

\begin{document}

\maketitle

\begin{abstract}
	In this paper we show that a generalized version of the Nikoli puzzle Slant is NP-complete. We also give polynomial time algorithms for versions of the puzzle where some constraints are omitted. These problems correspond to simultaneously satisfying connectivity and vertex degree constraints in a grid graph and its dual.
\end{abstract}

\section{Introduction}

Gokigen Naname, also known as Slant, is a pencil-and-paper logic puzzle by the Publisher Nikoli~\cite{Slant-website}. The puzzle involves filling every square in a grid with diagonal lines so they do not create cycles and each circle has the indicated number of diagonals touching it.
See \cref{fig:basic-example-board-and-solution} for an
example of a Slant puzzle.

\begin{figure}
    \centering
    \includegraphics[scale=1.1,page=1]{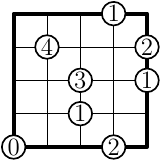}
    \hspace{1.5em}
    \includegraphics[scale=1.1,page=2]{basic-example-board}
    \caption{A simple example of a Slant puzzle, and its (unique) solution.}
    \label{fig:basic-example-board-and-solution}
\end{figure}

\paragraph{Rules}
Slant is played on a square grid, some of whose vertices may be given numbers (henceforth called a \defn{Slant board}).
The objective is to add a single diagonal line to each square obeying the following constraints:
\begin{enumerate}[noitemsep,topsep=0pt,parsep=0pt,partopsep=0pt]
    \item There is no cycle formed by diagonal lines.
    \item If a vertex has a designated number $k\in\{0,1,2,3,4\}$, the number of diagonal lines adjacent to the vertex must be exactly $k$.
\end{enumerate}
These shall henceforth be referred to as the \defn{cycle constraint}
and \defn{vertex constraints}.
Examples of solutions violating each can be seen in
\cref{fig:constraint-violation-examples}.

\begin{figure}[b]
    \centering
    \includegraphics[scale=1.1,page=2]{constraint-violation-examples}
    \hspace{3em}
    \includegraphics[scale=1.1,page=1]{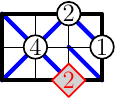}
    \caption{An example of each class of constraint violation denoted in red
        (the cycle constraint (left) and a vertex constraint (right)).}
    \label{fig:constraint-violation-examples}
\end{figure}

\paragraph{Results}
In \cref{sec:Preliminary} we discuss some combinatorial properties of the puzzle which will be used in our algorithmic and computational complexity results.

In \cref{sec:Polynomial Time} we give multiple algorithmic results for special cases of the Slant problem. We show that Slant is fixed-parameter tractable in the number of vertex constraints in the puzzle. We give an algorithm for deciding if a partially filled board can be completed without violating the cycle constraint. Finally we show how to formulate the vertex constraints and the cycle constraint as a matroid intersection problem allowing us to solve instances of the puzzle if the vertex constraints stay within certain partitions of the vertices.

Finally, in \cref{sec:NP-completeness}, we show that solving an instance of Slant on an $n\times m$ grid is NP-complete via a reduction from Hamiltonian cycle in max-degree-3 planar graphs.

\begin{figure}[t]
    \centering
    \includegraphics[scale=.75,page=1]{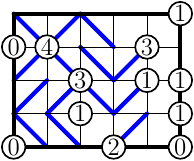}
    \hspace{0.25em}
    \includegraphics[scale=.75,page=2]{partial-board-fill-examples}
    \hspace{0.25em}
    \includegraphics[scale=.75,page=3]{partial-board-fill-examples}
    \caption{Two examples of partially-filled boards (left, middle).
        The leftmost board has no valid extension,
        but the middle board has a unique one, shown on the right.}
    \label{fig:partial-board-fill-examples}
\end{figure}

\paragraph{Related Work}
Many pencil-and-paper puzzles have been studied from the lens of computational complexity, especially those by the designer Nikoli. For example the following puzzles all have results showing generalized versions to be NP-complete:
Angle Loop~\cite{tang2022framework},
Bag / Corral~\cite{Friedman-2002-corral},
Chained Block~\cite{iwamoto2024chained},
Country Road~\cite{YajilinandCountryRoad},
Fillomino~\cite{Yato-2003},
Five-Cells~\cite{iwamoto2022five},
Hashiwokakero~\cite{Hashiwokakero},
Hebi~\cite{kanehiro2015satogaeri},
Heyawake~\cite{Holzer-Ruepp-2007},
Hiroimono / Goishi Hiroi~\cite{Andersson-2007},
Kouchoku~\cite{tang2022framework},
Kurodoko~\cite{Kurodoko},
Light Up / Akari~\cite{McPhail-2005-LightUp},
LITS~\cite{McPhail-2007-LITS},
Masyu / Pearl~\cite{Friedman-2002-pearl},
Mid-Loop~\cite{tang2022framework},
 Nagareru Loop~\cite{iwamoto2022moon},
 Nagenawa~\cite{tang2022framework},
 Nurimeizu\cite{iwamoto2022moon},
Numberlink~\cite{NumberlinkNP},
Nurikabe~\cite{McPhail-2003,Holzer-Klein-Kutrib-2004},
Ring-ring~\cite{tang2022framework},
Satogaeri~\cite{kanehiro2015satogaeri},
Shakashaka~\cite{Shakashaka,NumberlessShakashaka_CCCG2015},
Slitherlink~\cite{Yato-Seta-2003,Yato-2003,Witness_FUN2018},
Spiral Galaxies / Tentai Show~\cite{SpiralGalaxies, demaine2023rectangular},
Suraromu~\cite{kanehiro2015satogaeri},
Tatamibari~\cite{adler2020tatamibari},
Yajilin~\cite{YajilinandCountryRoad},
and
Yosenabe~\cite{Yosenabe}. There have also been multiple surveys on the topic of computational complexity and games and puzzles~\cite{uehara2023computational,kendall2008survey,demaine2001playing}.

Previous work showed NP-hardness for Yin-Yang puzzles and connected it to the problem of partitioning vertices into two trees with some vertices being pre-assigned\cite{demaine2021yin}. We find that Slant is also equivalent to a fairly natural graph problem in which we want to partition edges between a graph and its dual such that each forms a single tree with some vertices having degree constraints. This is discussed in more detail in Section~\ref{sec:Preliminary}.

\section{Preliminary Results and Key Observations}
\label{sec:Preliminary}

We discuss some preliminary results about Slant that will be used in other parts of the paper.

First, we observe that the two choices of diagonals have a parity constraint.
Each possible diagonal connects two vertices which differ by exactly one in both their $x$ and $y$-coordinates.
Thus, we can partition the grid into points whose coordinate-sum ($x+y$) is even,
and points whose coordinate-sum is odd.
Every square then admits a diagonal connecting two odd-sum points or a diagonal connecting two even-sum points.
Furthermore, if we rotate all of our coordinates by 45 degrees
it becomes obvious that every such diagonal belongs to one of two square grids graphs
(see \cref{fig:basic-example-board-45-rotation}).
Moreover, the two square grid graphs are essentially planar duals of each other.
That is, the faces of one correspond to the vertices of the other, and vice versa, with the same edge-vertex and edge-face incidences (with some exceptions for the boundary).
Thus, choosing a diagonal involves deciding to place an edge either in a primal square grid graph or a dual square grid graph.

\begin{figure}
    \centering
    \includegraphics[scale=0.52,page=4]{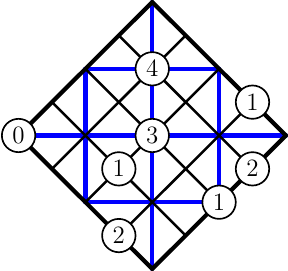}
    \hspace{0.15em}
    \includegraphics[scale=0.52,page=5]{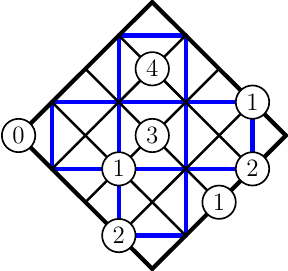}
    \hspace{0.15em}
    \includegraphics[scale=0.52,page=6]{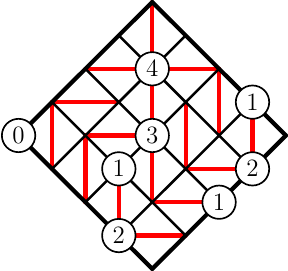}
    \caption{An example of how the Slant puzzle
    from
    \cref{fig:basic-example-board-and-solution}
    can be rotated $45^\circ$ so that all its potential edges come from two complementary grid graphs.}
    \label{fig:basic-example-board-45-rotation}
\end{figure}

With this dual graph perspective,
we can see that the acyclic constraint also implies that both graphs form forests where every tree touches the boundary.
By adding edges connecting the various points at which the boundary is touched, the complementary forests can be turned into a tree-cotree structure.
This also means that each parity class must contain roughly half the edges, in particular $(n-1)(m-1)/2+O(n+m)$ edges and the average degree is also $2+O(\frac{n+m}{nm})$.

\section{Polynomial Time Algorithms for Special Cases}
\label{sec:Polynomial Time}

We present variants of Slant that can be solved in polynomial time which use subclasses of the constraints.
The most important result in this section is a positive result on partially-filled boards:

\begin{theorem}
\label{thm:extend-partial-fill}
Given a partially-filled Slant board with no vertex constraints, whose filled-in diagonals do not already form a cycle, there always exists a valid solution extending the partially-filled in diagonals.
\end{theorem}

\begin{figure}
    \centering
    \includegraphics[scale=0.7,page=1]{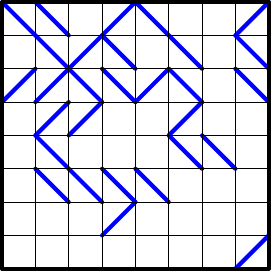}
    \hspace{1.5em}
    \includegraphics[scale=0.7,page=2]{blank-partial-board-fill-example}
    \caption{An example of a partial solution with no cycles or constraints, and a valid extension.}
    \label{fig:partial-fill-example}
\end{figure}

See \cref{fig:partial-fill-example} for an example.
Note that there is no difference between a partially-filled board with no vertex constraints,
and one whose vertex constraints all have their incident cells filled.
However, the lack of vertex constraints is necessary for this universality
(see \cref{fig:partial-board-fill-examples}).

\begin{proof}
Iterate through the unfilled squares in any order and make a choice of diagonal arbitrarily.
If that choice induces a cycle, choose the other diagonal instead. These two choices of diagonals connect vertices of different parities (odd degree vs even degree coordinate-sums),
and thus any curve they are a part of do not share vertices. Further, each pair of vertices is on opposite sides of the diagonal. Thus by required planarity of the connections and the Jordan curve theorem\footnote{Here we can use the simpler result that a simple polygonal curve divides the plane into two regions}\cite{thomassen1992jordan}, the other diagonal will not induce a cycle
(see \cref{fig:planarity-implies-greedy-fill} and \cref{fig:planarity-implies-greedy-fill-concrete}).
\end{proof}

\begin{figure}
    \centering
    \includegraphics[scale=0.9,page=1]{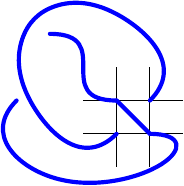}
    \hspace{3em}
    \includegraphics[scale=0.9,page=2]{planarity-implies-greedy-fill}
    \caption{A general demonstration of why if one configuration of a square (left) induces a cycle,
    the other (right) cannot.}
    \label{fig:planarity-implies-greedy-fill}
\end{figure}

\begin{figure}
    \centering
    \includegraphics[scale=0.85,page=4]{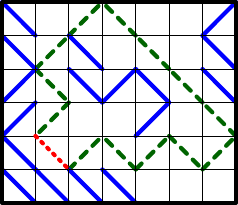}
    \hspace{1.5em}
    \includegraphics[scale=0.85,page=5]{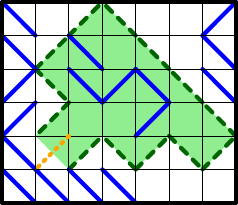}
    \caption{A concrete example of why if one configuration of a square (left) induces a cycle,
    the other (right) cannot (complementary to \cref{fig:planarity-implies-greedy-fill}).}
    \label{fig:planarity-implies-greedy-fill-concrete}
\end{figure}

\Cref{thm:extend-partial-fill}
will be particularly useful
as a component of the NP-hardness result in \cref{sec:NP-completeness},
and we will also use it to prove results in this section.

\begin{corollary}
\label{cor:4-and-0}
Given a Slant board whose designated numbers are all either $0$ or $4$,
it can be checked in polynomial time if the board permits a valid solution.
\end{corollary}

\begin{proof}
Every vertex constraint whose designated number is $0$ or $4$ requires
the same set of diagonals in any valid solution. In particular, degree-$4$ constrained vertices constraints require all four edges to be adjacent, and $0$ vertex constraints can only appear on the boundary (otherwise they force a cycle) and all edges to not be adjacent.
If there are any conflicts between these requirements,
or if they induce a cycle,
then the board does not permit a valid solution.
Otherwise, generate the corresponding partially-filled Slant board.
See \cref{fig:4-and-0-examples} for examples of each case.
If the board contains a cycle,
then it contains a cycle in any solution,
so it does not permit a valid solution.
Otherwise, apply \cref{thm:extend-partial-fill} to obtain a valid extension to the partially-filled board.
\end{proof}

\begin{figure}
    \centering
    \includegraphics[scale=1.1,page=1]{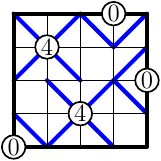}
    \hspace{1em}
    \includegraphics[scale=1.1,page=2]{4-and-0-examples}
    \caption{A board with designated number values of only $4$ and $0$ can generate either a valid partially-filled board generated (left),
    or induce cycle/constraint conflicts (right).}
    \label{fig:4-and-0-examples}
\end{figure}

Note that, via the same argument, this corollary can be extended to also allow designated numbers of $2$ along the boundaries,
and $1$ at corners (see \cref{fig:4-and-0-examples-extended}).

\begin{figure}
    \centering
    \includegraphics[scale=1.1,page=3]{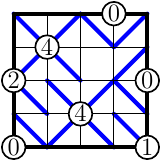}
    \caption{\cref{cor:4-and-0} extends to some special placements of $1$s and $2$s.}
    \label{fig:4-and-0-examples-extended}
\end{figure}

\subsection{Dense and Sparse Representations}

The decision variant of the Slant problem (whether or not a valid solution exists) is trivially in NP, but only if it is assumed that $\Theta(nm)$ values are used
to represent the full $n\times m$ grid (the \defn{dense representation}).
However, since not every vertex is required to have a designated number,
one could instead represent an instance of Slant with only the grid size,
and a list of non-empty coordinates with designated numbers (the \defn{sparse representation}).
In this representation, the size of the input may be asymptotically much smaller than $n\times m$.
We find that this does not change the complexity class:

\begin{theorem}
For an $n\times m$ Slant board with $k$ vertex constraints, using a sparse representation,
it can be checked if a valid solution exists in non-deterministic $\text{poly}(k\log(n+m))$ time.
\end{theorem}

\begin{proof}
For a certificate, consider a sparse repsentation of a partially-filled Slant board
whose filled squares are exactly those adjacent to any vertex with a designated number.
After non-deterministically assigning values to such a certificate,
it can be checked whether the designated numbers are satisfied,
and whether there are any cycles present among the assigned diagonals in $O(k)$ time.
If all designated numbers are satisfied by the partially-filled solution,
and no cycles are present,
by 
\cref{thm:extend-partial-fill},
there exists some valid extension of the partially-filled Slant grid,
and hence the answer to the decision problem is ``YES''.
\end{proof}

This same argument also implies that Slant is Fixed-Parameter Tractable in the number of vertex constraints. As before, we can exhaustively enumerate all of the possible configurations around the $k$ clues and then use \cref{thm:extend-partial-fill} to check for a valid extension.

\begin{corollary}
    Slant can be solved in time $O(16^k(n+m))$ where $k$ is the number of vertex clues, and is thus FPT in the number of vertex clues.
\end{corollary}

\subsection{Matroid Structure of Slant}

In this subsection, we show that Slant can be formulated as the intersection of five matroids.
While this on its own is not particularly surprising
(Slant is clearly in NP, and matroid intersection for $\geq3$ matroids
is known to be NP-hard~\cite{papadimitriou1998combinatorial}),
the particular structure of the matroids
will give us insight into some relaxations of Slant's constraints
that are polynomially-time solvable.
We will specifically show that Slant can be formulated as a weighted intersection of
four partition matroids, and one planar matroid.

A partition matroid is induced by a partition $\{P_i\}_i$
of its ground set $E$, and a mapping $P_i\mapsto b_i\in\mathbb{Z}_{\geq0}$
so that a set $S\subset E$ is independent if and only if $\left|S\cap P_i\right|\leq b_i$.

A planar matroid is one that is graphic and co-graphic.
That is, its basis is the set of maximal forests in some graph
(i.e., its independent sets are the forests in that graph),
and the basis of its dual is also the set of maximal forests in some graph.
Equivalently, the graph defining the graphic matroid is a planar graph.

First, we will observe that for an $n\times n$ Slant board,
the vertex constraints alone (without the cycle constraint) can be solved with an $0-1$ integer linear programming problem.
This will be done by creating a variable for each square indicating the choice of diagonal,
and a constraint for each vertex with a designated number.
The details of the construction are as follows:
\begin{itemize}[noitemsep,topsep=0pt,parsep=0pt,partopsep=0pt]
\item Use coordinates $(0,0)$ for both the top-left vertex of the grid, and for the top-left square of the grid.
\item For each square with coordinates $(x,y)$, create a variable $s_{x,y}\in\{0,1\}$.
\item For each vertex with coordinates $(x,y)$ whose designated number is $k\in\{0,1,2,3,4\}$, add an equality constraint:
        If $x+y$ is even, require that $\sum_{x_\Delta,y_\Delta\in\{-1,0\}}s_{x+x_\Delta,y+y_\Delta}=k$.
        Else, require that $\sum_{x_\Delta,y_\Delta\in\{-1,0\}}\left(1-s_{x+x_\Delta,y+y_\Delta}\right)=k$.
        Equivalently, that $\sum_{x_\Delta,y_\Delta\in\{-1,0\}}\left(s_{x+x_\Delta,y+y_\Delta}\right)=4-k$.
\end{itemize}
Given an assignment $s$ satisfying these constraints,
a corresponding solution to the Slant puzzle can be constructed as follows:
For a square whose coordinates are $(x,y)$, if $s_{x,y}=1$,
draw a diagonal between the two incident vertices
whose coordinate-sums are even.
Otherwise, draw a diagonal between the two incident vertices
whose coordinate-sums are odd.

We can use this ILP formulation to also obtain a matroid intersection formulation.
Let $E$ be the set of cells. $E$ will form the ground set for all five of our matroids.
Essentially, we will consider subsets $S\subset E$, and map these to possible solutions to a Slant instance.
A subset $S$ will represent the set of edges connecting incident vertices
with even coordinate-sums,
and likewise the set $E\setminus S$ will represent the set of edges
connecting incident vertices with odd coordinate-sums.
In this sense, the variables $s_{x,y}$ in the ILP formulation can be thought of as
a indicator functions for some subset $S\subset E$.

Consider the graph $G$ formed by connecting all vertices with even coordinate-sums.
This graph is bipartite, so call its parts $A$ and $B$.
A Slant instance induces a partition matroid over $E(G)$ for each of $A$ and $B$, which we will denote as $\mathcal L_A$ and
$\mathcal L_B$:
The independent sets of $\mathcal L_A$ are the subsets of edges which do not surpass the specified values at the vertices in $A$.
Similarly, the independent sets of $\mathcal L_B$ are the subsets of edges which do not surpass the specified values at the vertices in $B$.
In this case, $E(G)=E$ as defined before, so these are partition matroids over $E$, as desired.

We can also form a graph $G'$ formed by connecting all vertices with odd coordinate-sums.
Call its parts $C$ and $D$.
We define partition matroids over $E(G')$ in a symmetric manner,
except that in place of a specified value $k$ at a vertex in $C$ or $D$,
we use $4-k$.
Let the resulting partition matroids be denoted as $\mathcal L_C$ and $\mathcal L_D$.

With this formulation, the above ILP, encapsulating all but the cycle constraint,
can be solved via a $4$-way matroid intersection:
If a set $S$ is an independent set in all of
$\mathcal L_A$, $\mathcal L_B$, $\mathcal L_C$, and $\mathcal L_D$,
then its corresponding indicator $s$ function does not violate
relaxed constraints of the form:
$\sum_{x_\Delta\in\{-1,0\}}\sum_{y_\Delta\in\{-1,0\}}s_{x+x_\Delta,y+y_\Delta}\leq k$ (resp.~$4-k$, if applicable).
We weight the elements of $E$
by summing the left-hand side of all such constraints,
and we let $N$ be the sum of all right-hand sides.
Let $S'$
be
an independent set in all four matroids,
such that
its indicator function $s'$
is of maximum weight.
Then, $s'$
has
total weight equal to $N$
if and only if
all such constraints must be tight,
and hence if and only if $S'$ satisfies all non-cycle constraints of the Slant instance.

Lastly, we consider the cycle constraint,
which was not encapsulated in the ILP formulation.
A graph $G$ is cycle-free if and only if it is a forest.
This is encapsulated by a graphic matroid,
whose base set is $E(G)$.
In this case,
the graph we consider is that of the even coordinate-sum vertices.
This graph is planar and connected (in fact, it's a grid graph).
Hence, its basis consists of trees with $m$ edges for some value of $m$.
Call this matroid $\mathcal L_N$.
We modify the $4$-way matroid intersection to be a $5$-way matroid intersection as follows:
Add $1$ to the weight of every edge in $E$ as determined for the $4$-way intersection.
Let $S'$
be
an independent set in all five matroids,
such that
its indicator function $s'$
is of maximum weight.
Then, $s'$
has
total weight equal to $N+m-1$
if and only if
all constraints must be tight in all five matroids,
and hence if and only if $S'$ satisfies all constraints of the Slant instance.

Although three-way (and hence also four and five-way) matroid intersection
is NP-hard,
this formulation implies that
any pair of these constraints is solvable in polynomial time,
by the (two-way) matroid intersection theorem~\cite{papadimitriou1998combinatorial}.
For example, if we drop the cycle constraint, and only allow specified numbers
on vertices with even coordinate-sums (a checkerboard pattern),
then this relaxation of Slant becomes solvable in polynomial time.
This particular case also reduces to bipartite $b$-matching,
and even non-bipartite $b$-matching is well-known to be solvable in polynomial time without matroid intersection methods~\cite[Theorem 3.5.1]{marsh1979matching}.
The result can be more concisely summarized by the following pair of theorems:

\begin{theorem}
Let $a,b\in\{0,1\}$ be constants.
Suppose we are given a Slant board
such that each of its clue with coordinates $(x,y)$
has $x\equiv a\pmod 2$
and $y\equiv b\pmod 2$.
Then it can be checked in polynomial time
whether the board permits a solution
satisfying both the cycle and vertex constraints.
\end{theorem}

\begin{theorem}
Let $a_1,a_2,b_1,b_2\in\{0,1\}$ be constants.
Suppose we are given a Slant board
such that each of its clue with coordinates $(x,y)$
has either both
$x\equiv a_1\pmod 2$
and $y\equiv b_1\pmod 2$,
or both
$x\equiv a_2\pmod 2$
and $y\equiv b_2\pmod 2$.
Then it can be checked in polynomial time
whether the board permits a solution
ignoring the cycle constraint
and satisfying the vertex constraints.
\end{theorem}

\section{NP-completeness}
\label{sec:NP-completeness}

We will now prove that deciding if a Slant puzzle on an $n\times n$ board is NP-complete.
We use a two-step reduction:
We first review a (slightly modified) reduction from finding a Hamiltonian cycle with a single required edge in planar, bipartite, 3-regular graph
to finding a Hamiltonian cycle with a single required edge in a grid graph~
\cite{itai1982hamilton, munaro2017line} (this citation uses a slightly different Hamiltonian cycle variant, so we will add one additional gadget, and show that it is equivalent to their reduction).
Then, instead of using the latter problem for another black-box reduction,
we will show that we can replicate the functionality of each gadget used,
to form a two-step reduction.
In this section we will always be using the rotated view of Slant puzzles.

First, we will briefly summarize the key components of the grid graph reduction; readers interested in a full proof of this reduction's correctness should see~\cite{itai1982hamilton}. We then make two small modifications to the reduction. Next, we will give gadgets in Slant that force a portion of the graph to be a Hamiltonian cycle along a portion of the grid. Finally, we will address global connectivity issues and complete the proof.

\paragraph*{Grid graph reduction overview}
In Itai, Papadimitriou, and Szwarcfiter's reduction\cite{itai1982hamilton} of Hamiltonian cycle/path in grid graphs from Hamiltonian cycle in planar, bipartite, max-degree-$3$ graphs they take a grid embedding of the planar graph and construct edge and vertex gadgets. These are the vertices in \cref{fig:most-np-gadgets-phase-i} (note that the grid graph includes all unit-length edges, not just the example path going through the figure).
Call the planar graph $P$ and the grid graph $G$.
\begin{figure}
    \centering
    \includegraphics[scale=0.46,page=15]{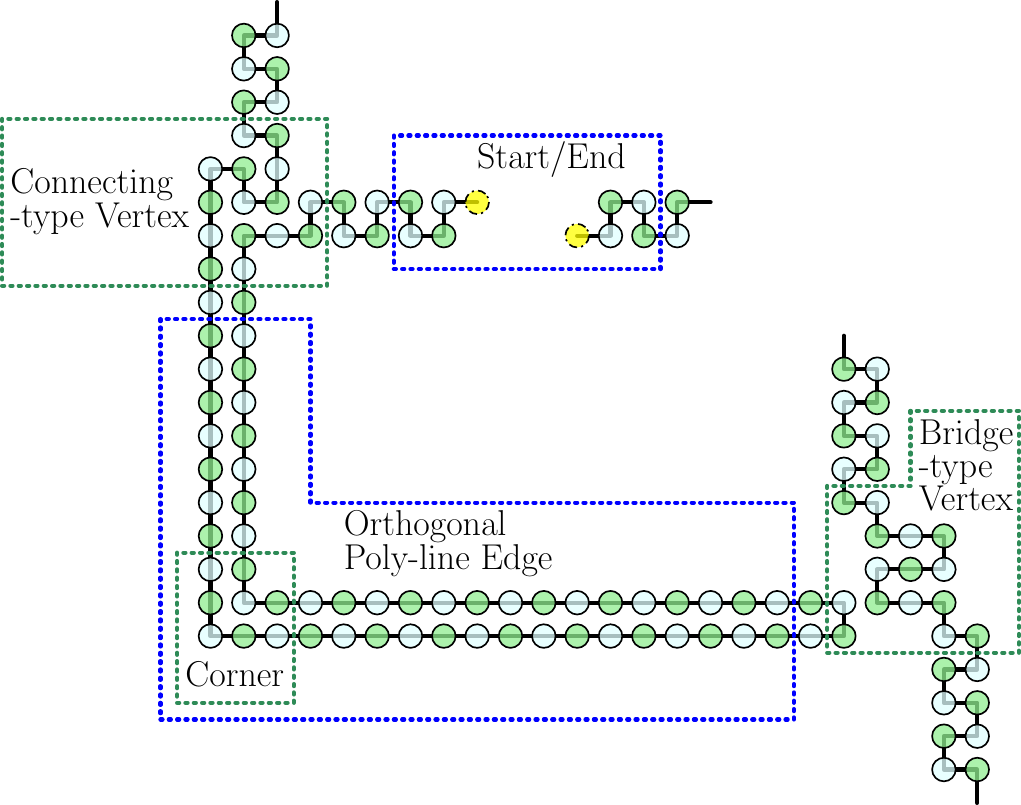}
    \caption{All of the gadgets used in the NP-Hardness construction for Hamiltonian path in grid graphs, and an example partial path going through them.}
    \label{fig:most-np-gadgets-phase-i}
    \label{fig:most-np-gadgets-phase-i-path}
\end{figure}

Edge gadgets are channels of vertices which are two vertices wide in $G$ and follow the layout of the embedding of $P$.
In a potential Hamiltonian cycle, they can be covered with either a zig-zag path (alternating pairs of left and right turns) or a $U$-shaped path.
A zig-zag path through $G$ has a connection at both ends of the edge gadget and corresponds to that edge gadget being in the Hamiltonian cycle in $P$.
A U-shaped path through $G$ has both ends at the same side of the edge gadget and corresponds to an edge which is not in the Hamiltonian cycle (however, we still need to cover all of the vertices in $G$).
Vertex gadgets are $3\times 3$ squares of vertices in $G$ and come in two types based on what parity class they belong to in $P$.
Vertex gadgets are split into two classes depending on the part of $P$ from which they were generated.
The difference between the classes lies in how they connect to their incident edge gadgets.
\cref{fig:most-np-gadgets-phase-i-path} provides a visualization of the potential paths and how they interact with the vertex gadgets.

Now we make two small modifications to this reduction. First, we want our graph to be $3$-regular, not just max-degree-$3$. Since this is a strict subproblem and already known to be NP-complete~\cite{akiyama1980np}, no additional work is needed. Next, we want the problem we are reducing from to have a specified edge which is required to be in the Hamiltonian cycle. This version of the problem is also known to be NP-complete~\cite{munaro2017line}, but this is not a strict subproblem. 
In our proof, we must also require a single forced-edge gadget.
Taking inspiration from Itai et al.'s~\cite{itai1982hamilton} use of degree-$1$ nodes to force the start and end of a Hamiltonian path,
we modify the required edge with the Start/End gadget shown in \cref{fig:most-np-gadgets-phase-i}. The degree-$1$ vertices in the Start/End gadget will force any Hamiltonian path in $G$ to start and end at those locations and the rest of the edge gadget to be filled with a zig-zag pattern (\cref{fig:most-np-gadgets-phase-i-path}).
With this property, the correctness of the reduction to Hamiltonian path in grid graphs using this gadget follows from the correctness of the Hamiltonian cycle reduction described by Itai et al.
We use this gadget instead of the approach of Itai et al.~because it will make the second phase of our two-step reduction more straightforward to describe and argue.

\begin{figure*}
    \begin{minipage}[b][][b]{.65\textwidth}
      \centering
        \includegraphics[scale=0.58,page=4]{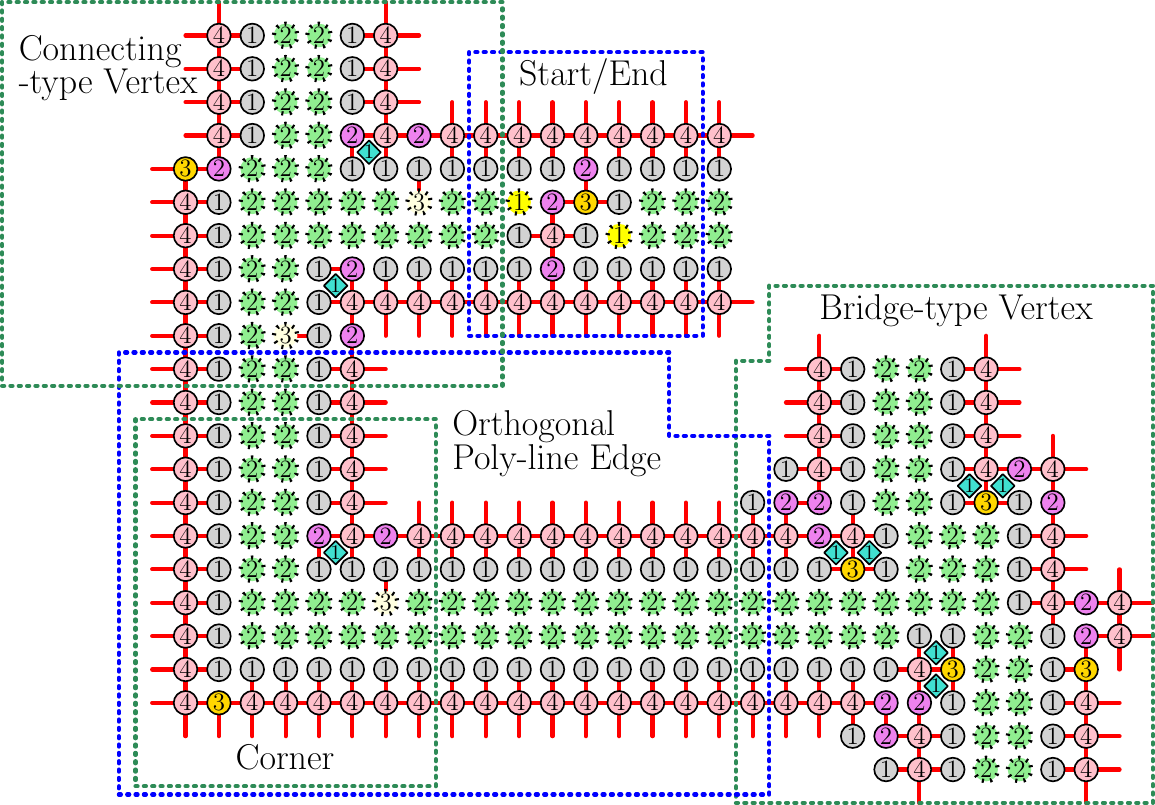}
        \caption{Most of the gadgets used in the NP-Hardness construction for Slant. Circles are vertices with their degree constraint written inside. Green, dotted circles are vertices in $S_G$. Diamonds are constrained vertices in the dual. Red lines are forced edges.}
            \label{fig:most-np-gadgets}
    \end{minipage}%
    \hfill
    \begin{minipage}[b][][b]{.3\textwidth}
      \centering
        \includegraphics[scale=0.9,page=4]{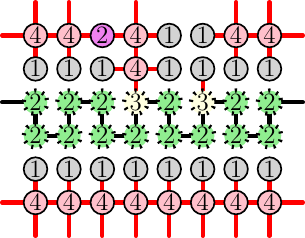} \\
        \vspace{1.75em}
        \includegraphics[scale=0.9,page=5]{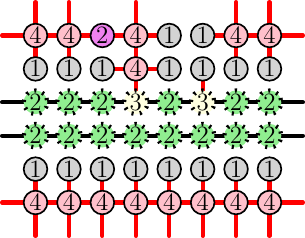}
        \caption{Face connection gadget in the NP-Hardness construction. Forced edges are in red.
        Two different configurations are given.}
        \label{fig:face-np-gadget}
    \end{minipage}
    
\end{figure*}

\paragraph*{Slant reduction}
With the above construction in mind we now turn to our second phase: how the grid graph reduction can be simulated in Slant.
That is, we will now construct a Slant puzzle.
Call the grid graph representing the even parity of Slant vertices $S$.
We will at all times assume the Slant puzzle is sufficiently large to admit our construction,
which will only require polynomial size.
We will first try to ensure that a portion $S_G$ of $S$ must admit a Hamiltonian path to properly satisfy that subset of the puzzle, and later we will discuss connecting the path to the rest of the puzzle to satisfy the cycle constraint of Slant.
Most importantly, we will construct portions of $S$ which act as barriers to $S_G$; it will not be possible to connect them to specific adjacent vertices, thus carving out a subset $S_G$ of $S$ that corresponds directly to $G$.
We will implement the grid vertices in $G$ with degree-$2$ constrained vertices in $S_G$, which cannot connect to any vertex in $S\setminus S_G$,
and some degree-$3$ constrained vertices that have exactly one forced connection to a grid point in $S\setminus S_G$ (this will be important to global connectivity, which we will discuss later),
and do not permit any others.
Additionally, there are exactly two degree-$1$ grid vertices in $G$,
which will correspond to degree-$1$ constrained vertices in $S_G$
that act as endpoints to any possible Hamiltonian path.
Any valid assignment of edges in $S_G$ satisfying these degree and the cycle constraint will be a Hamiltonian path on $S_G$, which itself will imply a Hamiltonian path through $G$.
Finally we will address the global connectivity issues in the puzzle that will allow us to satisfy all the constraints of the Slant puzzle if $G$ admits a Hamiltonian path and enforce that any assignment of $S$ that satisfies all the constraints of the Slant puzzle must have a Hamiltonian path over the vertices in $S_G$.

In order to reliably obtain the behaviour we want from $S_G$,
while admitting a Slant solution if $G$ contains a Hamiltonian path,
we will require Slant gadgets that admit each of the gadgets used by Itai et al..
All of the gadgets we will use in this reduction can be seen in \cref{fig:most-np-gadgets}.
There is a direct correspondence between the gadgets we use for Slant,
and the gadgets used in the initial grid graph reduction.
That is, we will constrain each part of $S_G$
corresponding to a gadget in $G$
so that its behaviour is the same in any valid solution to the Slant puzzle.
The main idea of the construction is to add degree-constrained vertices to $S$
that add \emph{forced} edges.
That is, in any valid solution to the final Slant instance we will construct,
these edges will be present.
Forced edges can be thought of as being possible to ``locally deduce'',
i.e.~determine that they are required in any valid solution using only a small neighbourhood around themselves.
Most of these forced edges rely on two useful degree-constraint properties to force edges to be in the primal or the dual. One,
degree-$4$ vertices must have all of their adjacent edges.
Two, if a degree constrained vertex already has as many adjacent forced primal edges as its constraint allows, then all other edges are forced to be in the dual (and vice versa).
The primary backbone of the edge gadgets is a sequence of degree-$4$ constrained vertices with adjacent degree-$1$ constrained vertices which are immediately saturated, preventing further connectivity.
In \cref{fig:most-np-gadgets},
all of the presented gadgets have their corresponding forced edges drawn.
Some of these edges require slightly more complex, though straightforward, local deductions:
\begin{itemize}[noitemsep,topsep=0pt,parsep=0pt,partopsep=0pt]
\item Degree-$1$ constrained vertices in the dual force three surrounding edges to be in the primal.
\item If two degree-$1$ constrained vertices have an edge between them they cannot have any other edges and will thus be disconnected from the rest of the graph (violating the cycle constraint in the dual).
\item
If
we have a pair of degree-$1$ constrained vertices next to a degree-$1$ constrained vertex in the dual then all three edges around that vertex in the dual are forced.
\item An edge is impossible if one of the endpoint vertices is already saturated.
\item A vertex diagonally across from a degree-$4$ constrained vertex cannot have two incident edges along the square shared with $4$, since doing so would guarantee the formation of a square (a cycle). Hence, if one of these edges is already present, the other is impossible.
\item For a given vertex $v$ with edge constraint $k$, if all but $k$ of its potential edges are impossible, then those edges are forced.
\end{itemize}
The above set of deduction rules is enough to obtain the entire set of edges
of the backbone
(i.e., the edges in \cref{fig:most-np-gadgets}).
In fact,
after adding all of the edges forced by degree-$4$ constrained vertices,
each of the remaining backbone edges at a vertex $v$
can be deduced from the constraints contained in the $3\times 3$ window
of vertices centred at $v$ using the above deductions.
Hence,
each deduction step can also be algorithmically implemented with a constant-time brute force check
looking at a small window of constraints on a partially-filled board.

\paragraph{Global connectivity.}
We have three classes of (primal) vertices in $S$:
The vertices in $S_G$,
the degree-specified vertices in $S\setminus S_G$ (i.e., the barrier),
and the unspecified vertices in $S\setminus S_G$ (i.e., the face interiors).
As specified, the degree-specified vertices in $S\setminus S_G$
will form a cycle around the boundary of each face in $P$ for any possible assignment of edges,
which would violate the cycle constraint in the primal.
Additionally, there are no connections between the first or second sets of vertices,
which would violate the cycle constraint in the dual.
We will solve both of these with a gadget simultaneously.
To fix this we will assign each face an edge and apply the edge-connection gadget, shown in Figure~\ref{fig:face-np-gadget}, to connect to that face. If the edge is not long enough to insert the gadget, we can simply scale up the embedding by an at-most polynomial factor.
In the connection gadget, the edge above the newly connected degree-$3$ constrained vertex is forced. Further this is the only place the face is able to connect to the rest of the construction. Thus the number of free edges at that vertex remains 2, the connectivity to the rest of the path is not changed, and thus this does not impact the admissible solutions in the edge.
Moreover, this gadget adds a single ``gap'' in the barrier of each face, removing the cycle.
Finally, for the remaining vertices in the puzzle (the unspecified vertices in $S\setminus S_G$), we leave them unconstrained which will allow the space outside our construction to be filled in using \cref{thm:extend-partial-fill}
for any valid partial solution to the vertices in $S_G$ and degree-specified vertices in $S\setminus S_G$.
There is no possible interaction between the vertices in $S_G$
and the unspecified vertices in $S\setminus S_G$, and hence this does not affect the ability of the vertices in $S_G$ to form a Hamiltonian path (corresponding to a Hamiltonian cycle with a specified forced edge in $P$).

Now, if the original 3-regular, bipartite, planar Hamiltonian cycle problem has a solution,
then so does the Hamiltonian path problem in $G$.
In our Slant instance, we can fill in the vertex and edge gadgets based on the solution in $G$, fill the faces as connected components with \cref{thm:extend-partial-fill}, and fill in the forced edges which gives a valid Slant solution (after rotating $45^\circ$).

If there is no solution to the Hamiltonian cycle problem in $P$, then we know there is no Hamiltonian path in the grid graph $G$. Our construction ensured that the green vertices which are in the same grid pattern at $G$ only have a solution if there is a Hamiltonian path among those vertices. Thus there would be no solution to the Slant puzzle, completing the reduction.

Combining the prior NP-hardness construction with the observation that checking a potential solution to a Slant puzzle can be done in polynomial time gives us our desired result.

\begin{theorem} \label{thm:np-complete}
Deciding if there is a solution to an $n\times n$ Slant puzzle is NP-complete.
\end{theorem}

\section{Conclusion and Open Questions}

In this paper we've connected the recreational logic puzzle Slant to the Hamiltonian path problem and matroid theory. Through these tools we've shown solving Slant puzzles is NP-complete and given algorithms to solve various special cases and simplifications of Slant.
In particular, we showed that the constraints of a Slant puzzle break into five classes,
and the simplifications that use just any two of them are solvable in polynomial time.
This combination of results leaves open the question of exactly how many/which classes of constraints are needed for NP-completeness.
Our construction currently uses all five classes of constraints (the cycle constraint, and all four classes of vertex constraints) and all types of vertex constraints except $0$.
It would be very interesting to know, for example, if the problem is still hard without the cycle constraint,
or if we only need the cycle constraint and vertex constraints in the primal graph
(since our construction makes very limited use of the dual graph).
We also know that puzzles which only contain $0$ and $4$ degree constraints are easy to solve, but what about other subsets?

Another category of questions deals with the uniqueness and quantity of solutions.
We can always find a way to fill a partially filled Slant puzzle with no remaining degree constraints,
but is there a polynomial time algorithm to count the number of solutions?
Further, since having a unique solution is a common design goal in pencil-and-paper puzzles,
it is natural to ask: is Slant ASP-hard~\cite{Yato-Seta-2003}?
Also, is counting the number of solutions \#P-hard\cite{valiant1979complexity}?

Finally, our view of Slant as a partition of edges between a planar and dual graph
naturally leads to a generalization of the puzzle to other types of graphs.
Is the problem still NP-hard on the triangular/hexagonal grid? In this case we consider the problem of partition edges between a primal and dual graph; as the interpretation of diagonals does not carry over to triangular girds.
Are there interesting and aesthetically pleasing similar puzzles on different forms of planar graphs?

\small
\bibliographystyle{abbrv}
\bibliography{main}

\end{document}